\documentclass[runningheads]{llncs}

\usepackage[utf8]{inputenc}
\usepackage{amsmath}
\usepackage{amssymb}
\usepackage{graphicx}
\usepackage{tabularx}
\usepackage{array}
\usepackage{pifont}
\usepackage{xspace}
\usepackage{todonotes}
\usepackage{paralist}
\usepackage{thm-restate}
\usepackage{subfigure}
\usepackage[pdfpagelabels,colorlinks,citecolor=blue,linkcolor=blue,urlcolor=blue]{hyperref}
\usepackage[english]{babel}
\usepackage{amsopn}
\usepackage{latexsym}
\usepackage{multirow}
\usepackage{multicol}
\usepackage{booktabs}
\usepackage[capitalise]{cleveref}
\usepackage{color, colortbl}
\usepackage{booktabs}

\newenvironment{cl}{Claim}

\graphicspath{{figures/}}

\definecolor{lightcyan}{rgb}{0.88,1,1}
\definecolor{antiquewhite}{rgb}{0.98, 0.92, 0.84}

\newcounter{casecounter}
\newcounter{subcasecounter}
\newcounter{subsubcasecounter}
\makeatletter
\newcommand{\ccase}[2][]{%
	\stepcounter{casecounter}%
	\setcounter{subcasecounter}{0}%
	\protected@write \@auxout {}{\string \newlabel {#2}{{#1\thecasecounter}{\thepage}{#1\thecasecounter}{#2}{}} }%
	\hypertarget{#2}{\noindent\textbf{Case #1\thecasecounter.}}
}

\newcommand{\subcase}[2][]{%
	\stepcounter{subcasecounter}%
	\setcounter{subsubcasecounter}{0}%
	\protected@write \@auxout {}{\string \newlabel {#2}{{#1\thecasecounter.\thesubcasecounter}{\thepage}{#1\thecasecounter.\thesubcasecounter}{#2}{}} }%
	\hypertarget{#2}{\noindent\textbf{Case #1\thecasecounter.\thesubcasecounter.}}
}

\newcommand{\subsubcase}[2][]{%
	\stepcounter{subsubcasecounter}%
	\protected@write \@auxout {}{\string \newlabel {#2}{{#1\thecasecounter.\thesubcasecounter.\thesubsubcasecounter}{\thepage}{#1\thecasecounter.\thesubcasecounter.\thesubsubcasecounter}{#2}{}} }%
	\hypertarget{#2}{\noindent\textbf{Case #1\thecasecounter.\thesubcasecounter.\thesubsubcasecounter.}}
}
\makeatother

\begin{document}
\title{Computing Weak Dominance Drawings with Minimum Number of Fips}
\author{Giacomo Ortali\inst{1}.
	Ioannis G. Tollis\inst{2}
}

\date{}

\institute{
	Universit\`a degli Studi di Perugia, Italy\\
	\email {giacomo.ortali@studenti.unipg.it}
	\and
	Computer Science Department, University of Crete, Heraklion, Crete, Greece \email{tollis@csd.uoc.gr}
}

\maketitle

%
\begin{abstract}
A weak dominance drawing $\Gamma$ of a DAG $G=(V,E)$,  is a $d$-dimensional drawing such that there is a directed path from a vertex $u$ to a vertex $v$ in $G$ if $D(u) <D(v)$ for every dimension $D$ of $\Gamma$. We have a \emph{falsely implied path (fip)} when $D(u) < D(v)$ for every dimension $D$ of~$\Gamma$, but there is no path from $u$ to $v$. 
Minimizing the number of fips is an important theoretical and practical problem, which is NP-hard. We show that it is an FPT~problem for parameter $k$, where $k$ is the maximum degree of a vertex of the \emph{modular~decomposition~tree} of~$G$. Namely, for any constant $d$, we present an $O(nm+ndk^2(k!)^d)$ time algorithm to compute a weak $d$-dimensional dominance drawing $\Gamma$ of a DAG $G$ having the minimum number of fips. An interesting implication of this result is that we can decide if a DAG has dominance dimension~$3$ (a well-known NP-complete problem) in time $O(nm+nk^2(k!)^3)$.

\end{abstract}
\noindent
\section{Introduction}
A \emph{directed acyclic graph} (DAG) $G=(V,E)$ is a directed graph with no directed cycles. 
For any dimension $D$ of a drawing $\Gamma$ of $G$, we denote by $D(v)$ the coordinate of vertex $v\in V$ in dimension $D$. A $d$-dimensional dominance drawing $\Gamma$ of $G$ is a $d$-dimensional drawing of $G$ where, given any pair of vertices $u,v\in V$, $D(u)< D(v)$ for every dimension $D$ of $\Gamma$ if and only if there exists a directed path connecting $u$ to $v$ in~$G$. We only consider directed paths. Hence, from now on we omit the word ``directed''. 
The efficient computation of dominance drawings of DAGs has many applications, including computational geometry~\cite{DBLP:conf/cccg/ElGindyHLMRW93}, graph drawing~\cite{st-planar},  databases~\cite{DBLP:conf/edbt/VelosoCJZ14}, etc. The \emph{dominance dimension} of $G$ is the minimum $d$ such that there exists a $d$-dimensional dominance drawing~of~$G$.

A \emph{partially order set (poset)} is a mathematical formalization of the concept of ordering. Any poset $P$ can be viewed as a transitive DAG $G^*$, i.e., as a DAG that contains its transitive closure graph. 
The dimension of $P$ is equivalent to the dominance dimension of any DAG $G$ whose transitive closure graph is $G^*$. 
The results obtained for DAGs and their dominance dimension transfer directly to posets and their dimension and vice versa (see~\cite{Yannakis} for a formal definition of poset and dimension of a poset). 
%
Due to this connection, we may assume that every dimension of our drawings is a topological order of the vertices of the graph. The literature concerning DAGs and their dominance dimension is vast. We report here some previous results.

Testing if a DAG has dominance dimension 2 requires linear time~\cite{Partiallyorderedsets,DBLP:conf/soda/McConnellS97}, while it is NP-complete to decide if the dominance dimension is greater than or equal to $3$~\cite{Yannakis}. A linear-time algorithm that constructs 2-dimensional dominance drawings of upward planar graphs is described in~\cite{DBLP:journals/dcg/BattistaTT92} (see also~\cite{st-planar}). The dominance dimension of a DAG with $n$ vertices is bounded~by~$\frac{n}{2}$~\cite{BOGART197321,N/2}.

Most DAGs have dominance dimension higher that two and, in general, computing dominance drawings with a bounded number of dimension is difficult. For this reason, a relaxed version of the concept of dominance drawings, the \emph{weak dominance drawings},  was introduced in~\cite{DBLP:conf/gd/KornaropoulosT12a}. In weak dominance, the ``if and only if'' of the definition of dominance becomes an ``if''. More formally, in a  weak dominance drawing $\Gamma$ of a DAG $G=(V,E)$, for any two vertices $u,v\in V$ there is a path from $u$ to $v$ in $G$ if $D(u) <D(v)$ for any dimension $D$ of $\Gamma$. We have a \emph{falsely implied path (fip)} when $D(u) < D(v)$ for every dimension $D$ of $\Gamma$, but $u$ and $v$ are incomparable, i.e., there is no path from $u$ to $v$ or from $v$ to $u$ in $G$.

For any DAG $G$ and any value $d$, $G$ admits a $d$-dimensional weak dominance drawing.  Given a $d$-dimensional weak dominance drawing, it is possible to test in $O(d)$ time if two vertices $u$ and $v$ of $G$ are incomparable. Otherwise, other kind of computations are required to check the existence of such a path, for example, a Breadth First Search (BFS), which would take $O(n + m)$ time per search.  Thus, minimizing the number of fips also minimizes the number of BFS potentially required and, consequently, it is also an important practical problem. 
However, the problem of minimizing the number of fips is NP-hard~\cite{DBLP:journals/corr/abs-1108-1439,DBLP:conf/gd/KornaropoulosT12a}.

Recently, the concept of weak dominance drawing was adopted in order to construct compact representations of  the reachability information of large graphs that are produced by large datasets in the database community, also considering high dimensional dominance drawings~\cite{DBLP:journals/www/LiHZ17,DBLP:conf/edbt/VelosoCJZ14}.
The number of fips (or false positives in their terminology) plays a crucial role.



\smallskip
\noindent \textbf{Our contribution:} We show that, for any constant $d$, computing a $d$-dimensional weak dominance drawing of a DAG $G$  with the minimum number of fips  is a fixed-parameter tractable (FPT) problem for parameter $k$, where $k$ is the maximum degree of a vertex of a \emph{modular~decomposition~tree}~of~$G$. Namely, for any constant $d$, we present an $O(nm+ndk^2(k!)^d)$ time algorithm to compute a $d$-dimensional weak dominance drawing $\Gamma$ of a DAG $G$ having the minimum number of fips. This result has interesting implications, for example, we can decide if a DAG has dominance dimension $3$ (a well-known NP-complete problem) in time $O(nm+nk^2(k!)^3)$. Similarly, for any constant number of dimensions.

\section{Preliminaries}
\label{se:preliminaries}
In this section we introduce two concepts that we use in the rest of the paper. In order to prove Theorem~\ref{th:fpt}, which is our main contribution, our strategy is to iteratively consider graphs obtained from the input DAG $G$ by merging some of its vertices into ``super-vertices''. This way of merging and the parameter $k$, that we use in our fixed-parameter algorithm, are introduced in Section~\ref{subse:moduladecompositiontree}, where we discuss the concept of \emph{modular decomposition tree}. In our algorithm every super-vertex has a \emph{cost}, that is equal to the number of vertices of $G$ merged to it. In Section~\ref{subse:costminimum} we define the concept of \emph{cost-minimum weak dominance drawing}. 

\subsection{Modular Decomposition Tree}
\label{subse:moduladecompositiontree}
Let $G=(V,E)$ be a DAG. An \emph{edge-based module} $M$ of $G$ is a subset of $V$ so that every vertex of $M$ is adjacent to the same set of vertices of $V\setminus M$. More formally, either $|M|=1$ or, for any two vertices $v_1,v_2\in M$ and any vertex $u\in V\setminus M$:  $(v_1,u)\in E$ if and only if $(v_2,u)\in E$; $(u,v_1)\in E$ if and only if $(u,v_2)\in E$. In the literature, the edge-based modules are simply called ``modules''~\cite{DBLP:journals/dm/McConnellS99}.

The \emph{edge-based congruence partition} $C_P$ of $V$ is a partition of $V$ into edge-based modules. The \emph{edge-based quotient graph} $G/C_P$ is the graph obtained from $G$ by merging into a ``super-vertex'' the vertices of each edge-based module of $C_P$. The \emph{edge-based modular decomposition} of $G$ is a tree $T$ describing a decomposition of $G$ based into its edge-based modules. The root of $T$ is the trivial edge-based  module $V$ and any leaf of $T$ is a trivial edge-based module~$\{v\}$, where $v\in V$.
For further details about the concepts defined so far~see~\cite{DBLP:journals/dm/McConnellS99}.

\begin{example}	\label{ex:edge-based-moddecomp} Fig.~\ref{fi:edge-modular-decompsition}(a) depicts a DAG $G$ and its edge-based modules. The non-trivial modules are $M_1=\{2,3\}$ and $M_2=\{6,11\}$. Fig.~\ref{fi:modular-decompsition}(b) depicts the edge-based modular decomposition tree of $G$.
\end{example}

In this paper we consider \emph{path-based modules}.  A path-based module $M$ of $G$ is a non-empty subset of $V$ so that either $|M|=1$ or, for any two vertices $v_1,v_2\in M$ and any vertex $u\in V\setminus M$: There is a path connecting $v_1$ to $u$ if and only if there is a path connecting $v_2$ to $u$; there is a path connecting $u$ to $v_1$ if and only if there is a path connecting $u$ to $v_2$.  The path-based modules are used in~\cite{DBLP:journals/dase/AnirbanWI19} as a generalization of the concept of edge-based modules.

Notice that a path-based module of $G$ is an edge-based module for the transitive closure graph $G^*$ of $G$. Hence, any result for the edge-based modules can be transferred to the path-based module. Therefore, we have the concepts of \emph{path-based congruence partition}, \emph{path-based quotient graph}, and \emph{path-based modular decomposition tree}. Since the edge-based modular decomposition tree can be computed in $O(m)$ time~\cite{DBLP:journals/dm/McConnellS99}, the path-based modular decomposition tree can be computed in $O(nm)$ time.

\begin{figure}[h]
	\centering
	\subfigure[]{\includegraphics[width=0.49\columnwidth,page=3]{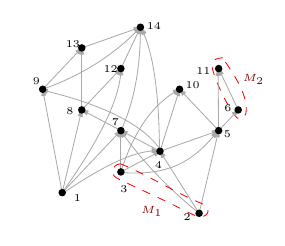}}
	\hfil
	\subfigure[]{\includegraphics[width=0.49\columnwidth,page=4]{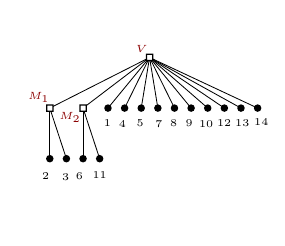}}
	\hfil
	\caption{A DAG $G$ and a the edge-based modular decomposition tree of $G$.}\label{fi:edge-modular-decompsition}
\end{figure}

\begin{figure}[h]
	\centering
	\subfigure[]{\includegraphics[width=0.49\columnwidth,page=1]{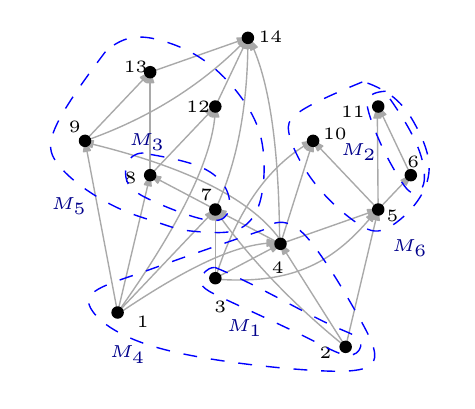}}
	\hfil
	\subfigure[]{\includegraphics[width=0.49\columnwidth,page=2]{modular-decomposition}}
	\hfil
	\caption{A DAG $G$ and the path-based modular decomposition tree of $G$.}\label{fi:modular-decompsition}
\end{figure}

\begin{example}
	\label{ex:path-based-moddecomp} Fig.~\ref{fi:modular-decompsition}(a) depicts a DAG $G$ and its path-based modules. Notice that $M_1$ and $M_2$ are also edge-based modules. Fig.~\ref{fi:modular-decompsition}(b) depicts the path-based modular decomposition tree of $G$. 
\end{example}

\begin{example}
	\label{ex:path-based-moddecomp-2}
	 Fig.~\ref{fi:modules}(a) depicts a DAG $G$ and a path-based congruence partition $C_P=\{M_1,M_2,M_3,$ $M_4,M_5\}$ of $G$. Fig.~\ref{fi:modules}(b) depicts the path-based quotient graph $G/C_P$, where every vertex $v_i$ is the super-vertex associated to the path-based module $M_i\in C/P$~($i\in [1,5]$).
\end{example}

The parameter that we consider for our fixed-parameter algorithm is the maximum degree $k$ of a vertex of the edge/path-based modular decomposition tree.  Since the definition of edge-based module is more restrictive than the definition of path-based module (i.e., an edge-based module is always a path-based module, but not vice versa), the value of $k$ for the path-based modular decomposition tree is less than or equal to the value of $k$ for the edge-based modular decomposition tree. For this reason, we consider only path-based modules.

\begin{example}
\label{ex:kvalues}
Consider the graph in Fig.~\ref{fi:edge-modular-decompsition}(a) (also depicted in Fig.~\ref{fi:modular-decompsition}(a)). If we consider its edge-based modular decomposition tree, depicted in Fig.~\ref{fi:edge-modular-decompsition}(b), we have $k=12$. If we consider its path-based modular decomposition tree, depicted in Fig.~\ref{fi:modular-decompsition}(b), we have $k=4$. 
\end{example}

Since here we only consider path-based modules, for simplicity we omit the word ``path-based''. Given the modular decomposition tree $T$ of $G$ and a module $M$, if $|M|>1$ we associate to $G_M$ the congruence partition induced by the children of the vertex associated to $M$ in $T$. Also, we denote by $G_M$ the subgraph of $G$ induced by the vertices in~$M$.  

\begin{example} Consider Fig.~\ref{fi:modular-decompsition}.  The root of the tree is the trivial module $V$ and $G_V=G$. We associate to $G$ the congruence parition $\{M_4,M_5,M_6,\{14\}\}$. We associate to $G_{M_4}$ the congruence partition $\{M_1,\{1\},\{4\}\}$.
\end{example}

\subsection{Cost-Minimum Weak Dominance Drawings}
\label{subse:costminimum}
Let $H$ be a DAG such that every vertex $v$ is assigned a \emph{cost} $c(v)$. Let $\Gamma$ be a weak dominance drawing of $H$.  The cost of a fip $(u,v)$ in $\Gamma$ is $c(u,v)=c(u)\cdot c(v)$. The \emph{cost} of $\Gamma$ is the sum of the costs of its fips.  Let $G$ be a DAG, $C_P=\{M_1,...M_h\}$ be a congruence partition of $G$, and $v_i$ be the super-vertex representing $M_i\in C_P$ in $G/C_P$ ($i\in [1,h]$). We assign the cost to the vertices of $G$ and $G/C_P$ as follows:
\begin{itemize}
	\item For any $v\in G$, $c(v)=1$.
	\item For any $v_i\in G/C_P$, $c(v_i)=|M_i|$ ($i\in [1,h]$).
\end{itemize}
With this cost assignment, the cost of a weak dominance drawing of $G$ is equal to its number of fips.

\begin{example}  
	\label{ex:cost} 
	Fig.~\ref{fi:modules}(a) depicts a DAG $G$ and a congruence partition $C_P=\{M_1,M_2,$ $M_3,$ $M_4,M_5\}$ of $G$. Fig.~\ref{fi:modules}(b) depicts the quotient graph $G/C_P$. Fig.~\ref{fi:modules}(c) depicts three 2-dimensional weak dominance drawings~of~$G/C_P$. 
	
	\noindent$\bullet$  Drawing $\Gamma_1$ has the following six fips: \emph{\#1}~$(v_1,v_2)$; \emph{\#2}~$(v_1,v_3)$; \emph{\#3}~$(v_2,v_5)$; \emph{\#4}~$(v_3,v_4)$; \emph{\#5}~$(v_4,v_5)$; \emph{\#6}~$(v_6,v_5)$. The cost of $\Gamma_1$ is  $c(v_1,v_2)+c(v_1,v_3)+c(v_2,v_5)+c(v_3,v_4)+c(v_4,v_5)+c(v_6,v_5)=6+12+1+18+9+1=47$. 
	
	\noindent$\bullet$ Drawing $\Gamma_2$ contains only the fip $(v_3,v_4)$ and its cost is $c(v_3,v_4)=18$. 
	
	\noindent$\bullet$ Drawing $\Gamma_3$ contains only the fip $(v_2,v_5)$ and its cost is $c(v_2,v_5)=1$. 
	
	\noindent
	Since $G/C_P$ is the crown graph, the cost of a weak dominance drawing of $G/C_P$ it is at least $1$. See, for example,~\cite{DBLP:conf/gd/KornaropoulosT12a} . Hence, $\Gamma_3$ is a cost-minimum 2-dimensional weak dominance drawing of~$G/C_P$.
\end{example}

\begin{figure}[h]
	\centering
	\subfigure[]{\includegraphics[width=0.57\columnwidth,page=1]{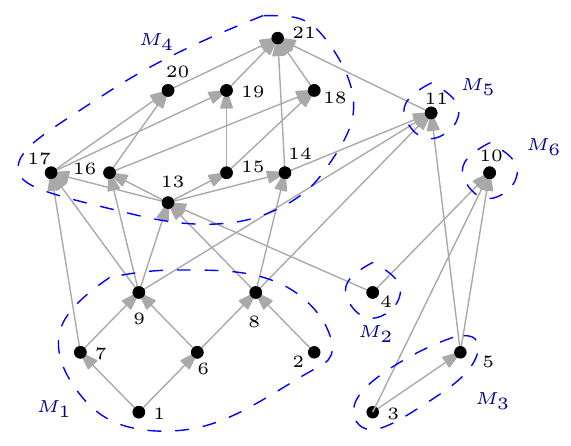}}
	\hfill
	\subfigure[]{\includegraphics[width=0.42\columnwidth,page=2]{min-fip-example}}
	\hfill
	\subfigure[]{	\includegraphics[width=0.7\columnwidth,page=8]{min-fip-example}}
	\hfill
	\caption{(a) A DAG $G$ and a congruence partition $C_P=\{M_1,...,M_6\}$. (b) The correspondent quotient graph $G/C_P$ and the costs of its vertices. (c) Three different 2-dimensional weak dominance drawings $\Gamma_1$, $\Gamma_2$, and $\Gamma_3$ of $G/C_P$. Drawing $\Gamma_3$ is cost-minimum.}\label{fi:modules}
\end{figure}

\section{Minimizing the Number of Fips}
\label{se:th1}
In this section we present the main contribution of this paper, stated in Theorem~\ref{th:fpt}. Before presenting and proving the theorem, we show two intermediate results: Lemma~\ref{le:compaction} and Lemma~\ref{le:cost_CP}.

A module $M$ of a DAG $G$ is \emph{compact} in a dimension $D$ of a weak dominance drawing $\Gamma$ of $G$ if the coordinates of the vertices of $M$ in  $D$ are consecutive. Also, $M$ is \emph{compact} in $\Gamma$ if it is compact in every dimension of $\Gamma$. A congruence partition $C_P$ of $G$ is compact in $\Gamma$ if every module $M\in C_P$ is compact in $\Gamma$. 

\begin{example}
Fig.~\ref{fi:compaction}(a-c) depicts three 2-dimensional weak dominance drawings of the same graph. Module $M=\{5,7,10,12,14\}$ is: not compact in any dimension in~(a); compact in Dimension $Y$ in~(b), but not in Dimension $X$; compact in~(c).
\end{example}

\begin{example} Fig.~\ref{fi:cost-minimumdrawings}(e) depicts a 2-dimensional weak dominance drawing where the congruence partition $C_P=\{M_1,...,M_6\}$ is compact. 
\end{example}

Recall that we assume $c(v)=1$ for every vertex $v$ of $G$ and that a cost-minimum weak dominance drawing of $G$ is a weak dominance drawing of $G$ with the minimum number of fips. The following lemma, which is proved in Section~\ref{se:compaction_proof}, is the first ingredient that we need before proving Theorem~\ref{th:fpt}.

\begin{restatable}{lemma}{compact}
	\label{le:compaction}
	Let $G$ be a DAG, $C_P$ be a congruence partition of $G$ and $d$ be a constant. There exists a cost-minimum $d$-dimensional weak dominance drawing of $G$ where $C_P$ is compact.
\end{restatable} 

Lemma~\ref{le:compaction} implies that it is possible to compute a cost-minimum $d$-dimensional weak dominance drawing of $G$ where $C_P$ is compact without loss of generality.
 
Given a DAG $G$ and a congruence partition $C_P$ of $G$, let $opt(G/C_P)$ be the cost of a cost-minimum weak dominance drawing of $G/C_P$. 
Given a weak dominance drawing $\Gamma$ of $G$ and $C_P$, we say that a fip $(u,v)$ of $\Gamma$ is an \emph{inner-fip} if $u,v\in M$, where $M$ is a module of $C_P$. Otherwise, $(u,v)$ is an \emph{outer-fip}. 

\begin{example} As we showed in Example~\ref{ex:cost}, $\Gamma_3$ in Fig.~\ref{fi:modules}(c) is a cost-minimum 2-dimensional weak dominance drawing of the graph $G/C_P$ in Fig.~\ref{fi:modules}(b) and its cost is $1$. Hence, $opt(G/C_P)=1$.
\end{example}

\begin{example}
	Consider the 2-dimensional weak dominance drawing in Fig.~\ref{fi:cost-minimumdrawings}(e) and the congruence partition $C_P=\{M_1...M_6\}$. Fip $(16,19)$ is an inner-fip. Fip $(4,11)$ is an outer-fip. 
\end{example}

\medskip
The next lemma relates the number of outer-fips of a weak dominance drawing of $G$ where $C_P$ is compact to the value~$opt(G/C_P)$.

\begin{lemma}
\label{le:cost_CP}
Let $G$ be a DAG and $C_P$ be a congruence partition of $G$. For any weak dominance drawing $\Gamma$ of $G$ such that $C_P$ is compact in~$\Gamma$ and having $t$ outer-fips, we have $t\ge opt(G/C_P)$.
\end{lemma}
\begin{proof}
	Let $C_P=\{M_1,...,M_h\}$. Let $\Gamma$ be a weak dominance drawing of $G$ where $C_P$ is compact and with $t$ outer-fips. It is possible to construct a weak dominance drawing $\Gamma'$ of $G/C_P$ by: Contracting every $M_i\in C_P$ ($i\in [1,h]$) to a vertex $v_i$ in $\Gamma$; assigning $c(v_i)=|M_i|$. See Example~\ref{ex:contraction}.
	Let $t'$ be the cost of $\Gamma'$. Notice that $t'\ge opt(G,C_P)$ by definition. We now~prove~$t'=t$. 
	Since $C_P$ is compact in $\Gamma$ and by definition of module, we have that: For every fip $(u,v)$ of $\Gamma$ such that $u\in M_i$ and $v\in M_j$, where $i,j\in [1,h]$ and $i\not =j$,  there is a fip $(u',v')$ for any couple of vertices $u'$ and $v'$ such that $u'\in M_i$ and $v'\in M_j$. Hence, any fip $(u,v)$ of $\Gamma$ implies the existence of $|M_i|\cdot|M_j|$ outer-fips in $\Gamma$. Also, since $\Gamma'$ is obtained by contracting the vertices of every module of $C_P$ and since $C_P$ is compact in $\Gamma$, fip $(u,v)$ implies a fip $(v_i,v_j)$ having a cost $c(v_i)\cdot c(v_j)=|M_i|\cdot|M_j|$ in $\Gamma'$. Hence, $t'\ge t$. By a symmetric argument, every fip $(v_i,v_j)$ in $\Gamma'$ implies the existence of   $|M_i|\cdot |M_j|$ outer-fips in $\Gamma$. Hence, $t'\le t$. Since $t'\ge t$ and $t'\le t$ we have $t'=t$. 
	Since $t'=t$ and $t'\ge opt(G,C_P)$, we have $t\ge opt(G,C_P)$.
\end{proof}

\begin{example}
	\label{ex:contraction} Fig.~\ref{fi:cost-minimumdrawings}(e) depicts a 2-dimensional weak dominance drawing $\Gamma$ of the graph in Fig.~\ref{fi:modules}(a) where the congruence partition $C_P=\{M_1,...,M_6\}$ is compact. By contracting every module $M_i$ to a vertex $v_i$ and by assigning $c(v_i)=|M_i|$ ($i\in [1,6]$) we obtain the weak dominance drawing of $G/C_P$ in Fig.~\ref{fi:cost-minimumdrawings}(d).
\end{example}
We are ready to prove Theorem~\ref{th:fpt}, which is our main contribution.
\begin{theorem}
	\label{th:fpt}
	Let $G=(V,E)$ be a DAG. Let $k$ be the maximum degree of a vertex of the modular decomposition tree $T$ of $G$. For any constant $d$, it is possible to compute a $d$-dimensional weak dominance drawing of $G$ with the minimum number of fips in $O(nm+ndk^2(k!)^d)$ time.
\end{theorem}
\begin{proof}
Recall that, since we assign $c(v)=1$ for every $v\in V$, for any module $M$ of $G$ a cost-minimum weak dominance drawing of $G_M$ has the minimum number of fips. If $M=V$, $G_V=G$. It is possible to compute a cost-minimum $d$-dimensional weak dominance drawing of any graph $H$ having $k$ vertices in $O(dk^2(k!)^d)$ time by using the following brute force algorithm: (a)~Compute all the possible $d$-dimensional weak dominance drawings of $H$ in $O((k!)^d)$ time; (b)~test in $O(dk^2)$ time, for each drawing, its cost; (c)~select the drawing with the minimum cost. In order to compute the cost-minimum $d$-dimensional weak dominance drawing of $G$ drawing we do a bottom-up traversal of $T$.

	\medskip\noindent \texttt{Base Step:} Let $M$ be a module of $G$ such that the children of the corresponding vertex in $T$ are leaves of $T$. Every module of the congruence partition $C_P$ associated to $G_{M}$ is a trivial module with cardinality 1. Hence, $G_{M}$ is equal to the quotient graph $G_{M}/C_P$ and it has less than $k$ vertices. It is possible to compute a cost-minimum weak dominance drawing of $G_{M}$ in $O(dk^2(k!)^d)$ time.

	\medskip\noindent \texttt{Recursive Step:} Let $M$ be a module of $G$ such that  the corresponding vertex in $T$ has $k'\le k$ children that are not all leaves of $T$ (i.e., one of them is an internal vertex of $T$). In order to simplify the notation, suppose $k'=k$ and $M=V$ without loss of generality. I.e., $G_{M}=G$. Let $C_P=\{M_1,...,M_h\}$ be the congruence partition that we associate to $G$ given~$T$ (see Example~\ref{ex:path-based-moddecomp-2}).
	By inductive hypothesis, the cost-minimum weak dominance drawing $\Gamma_{M_i}$ of $G_{M_i}$ is given for every $M_i\in C_P$ ($i\in [1,h]$). We now compute a cost-minimum $d$-dimensional weak dominance drawing $\Gamma$ of $G$. Since $G/C_P$ has $k$ vertices, we can compute a cost-minimum weak dominance drawing $\Gamma'$ of $G/C_P$ in  $O(dk^2(k!)^d)$ time. Recall that, for any $M_i\in C_P$ ($i\in [1,h]$), $c(v_i)=|M_i|$, where $v_i$ is the vertex of $G/C_P$ associated~to~$M_i$.
	
	\begin{example} 
		Consider the DAG $G$ in Fig.~\ref{fi:modules}(a) and the case $d=2$. For any $i\in \{1,2,3,5,6\}$, $G_{M_i}$ is planar and $\Gamma_{M_i}$ is a dominance drawing. Drawing $\Gamma_{M_1}$ is in Fig.~\ref{fi:cost-minimumdrawings}(a), while $\Gamma_{M_2}$, $\Gamma_{M_3}$, $\Gamma_{M_5}$, and $\Gamma_{M_6}$ are very simple and they are depicted in Fig.~\ref{fi:cost-minimumdrawings}(b). Fig.~\ref{fi:cost-minimumdrawings}(c) depicts $\Gamma_{M_4}$. We have that $G_{M_4}$ contains the crown graph and $\Gamma_{M_4}$ is cost-minimum, since it has one fip, that is $(16,19)$.  Fig.~\ref{fi:modules}(a) depicts $\Gamma'$, where $G/C_P$ is in Fig.~\ref{fi:modules}(b). Drawing $\Gamma'$ is $\Gamma_3$ of Fig.~\ref{fi:modules}(c), which is cost-minimum, as described in Example~\ref{ex:cost}.
	\end{example}
	\begin{figure}[h]
		\centering
		\subfigure[]{\includegraphics[width=0.32\columnwidth,page=4]{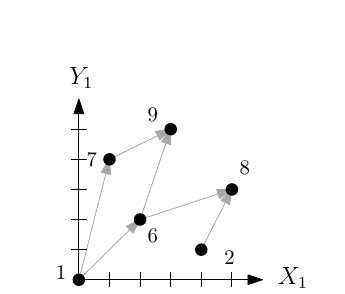}}
		\hfil
		\subfigure[]{\includegraphics[width=0.32\columnwidth,page=9]{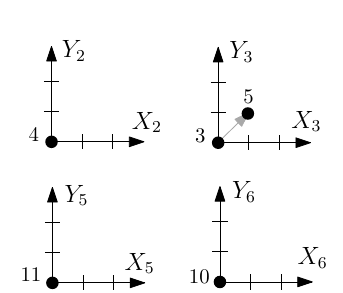}}
		\hfil
		\subfigure[]{\includegraphics[width=0.32\columnwidth,page=5]{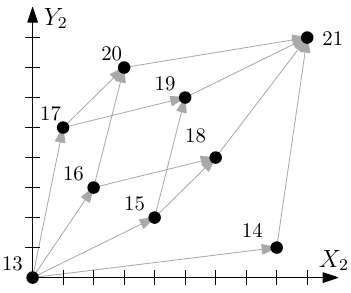}}
		\hfil
		\subfigure[]{\includegraphics[width=0.32\columnwidth,page=3]{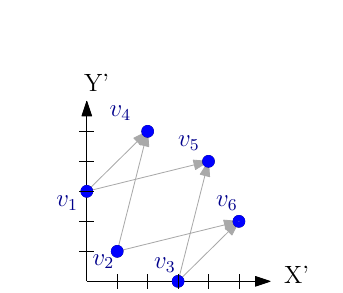}}
		\hfil
		\subfigure[]{\includegraphics[width=0.45\columnwidth,page=7]{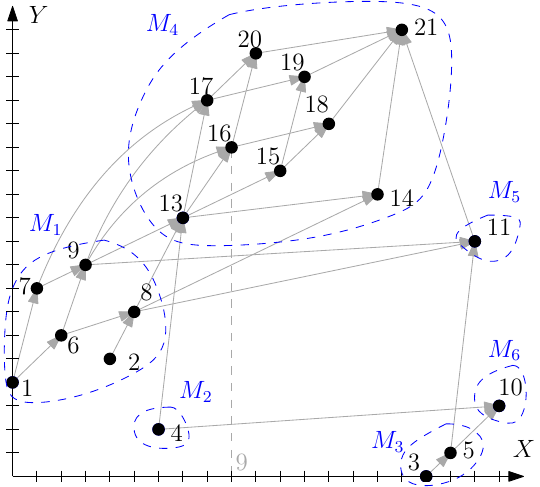}}
		\hfil
		\caption{Refer to $G$ and $C_P$ in Fig.~\ref{fi:modules}. The cost-minimum weak dominance drawings of: (a)~$G_{M_1}$; (b)~$G_{M_2}$, $G_{M_3}$, $G_{M_5}$, $G_{M_6}$; (c)~$G_{M_4}$; (d)~$G/C_P$; (e)~$G$.}\label{fi:cost-minimumdrawings}
	\end{figure}
	
	We compute a weak dominance drawing $\Gamma$ of $G$ by expanding the vertex $v_i$ to the drawing $\Gamma_{M_i}$ in $\Gamma'$, for any $i\in [1,h]$. See Figures~\ref{fi:cost-minimumdrawings}(d) and~(e). More formally, let $D'$ a dimension of $\Gamma'$ and let $V_{D'}^j$ be a set of vertices of $G/C_P$ so that $D'(v)<D'(v_j)$ for any $v\in V_{D'}^j$. We perform the following coordinate assignment operation for every $v\in G$ and for every dimension $D$~of~$\Gamma$:
	
	\smallskip
	\noindent
	\texttt{Coordinates Assignment Operation ($v$, $D$):} Let $M_i$ be the module of $C_P$ containing $v$ ($i\in [1,h]$). Let $D'$ and $D_i$ be the dimension corresponding to $D$ in $\Gamma'$ and $\Gamma_i$. I.e., if $D$ is the $g$th dimension of $\Gamma$, $D'$ and $D_i$ are the $g$th dimension of $\Gamma'$ and $\Gamma_i$, respectively. Set $D(v)=D'(v_i)+\sum_{u\in V_{D'}^j}(c(u)-1)+D_i(v)$.

	\begin{example} Refer to Fig.~\ref{fi:cost-minimumdrawings}. Fig.~\ref{fi:cost-minimumdrawings}(e) depicts $\Gamma$ after the Coordinates Assignment Operation is performed for every $v\in G$ and every dimension of~$\Gamma$. The graphs $G$ and $G/C_P$ are depicted in Fig.~\ref{fi:modules}(a) and~(b). Refer to Vertex $16$ and dimension $X$. We have $16\in M_4$ and $V_X^4=\{v_1,v_2\}$. Hence, $X(16)=X'(v_4)+(|c(v_1)|-1)+(|c(v_2)|-1)+X_4(16)=2+(6-1)+(1-1)+2=9$. 
	\end{example}

	Recall that, for $i\in [1,h]$, the vertices in $\Gamma_{M_i}$ are contained in the same module $M_i$ of $G$. Hence, since we obtain $\Gamma$ by expanding the vertices of $\Gamma'$ to the drawings $\Gamma_{M_1},...,\Gamma_{M_h}$ and since $\Gamma',\Gamma_{M_1},...,\Gamma_{M_h}$ are weak dominance drawing, we have that $\Gamma$ is a weak dominance drawing and that $C_P$ is compact in $\Gamma$.  We now show that $\Gamma$ is a cost-minimum $d$-dimensional dominance drawing of $G$ having the minimum number of fips among all the $d$-dimensional dominance drawings of $G$ where $C_P$ is compact. By Lemma~\ref{le:compaction}, it implies that $\Gamma$ is cost-minimum.

	\smallskip \noindent \emph{The inner-fips:} Notice that the drawing $\Gamma$ restricted to the vertices of any module $M_i\in G/C_P$ is $\Gamma_{M_i}$. Since $\Gamma_{M_i}$ is cost-minimum, we have that $\Gamma$ has the minimum number of inner-fips. 
	
	\smallskip \noindent \emph{The outer-fips:} Notice that $\Gamma'$ has a cost $opt(G/C_P)$. Since $G/C_P$ is compact in $\Gamma$ and since we obtained $\Gamma$ by expanding the vertices of $\Gamma'$ to the drawings $\Gamma_{M_1},...,\Gamma_{M_h}$, we can prove by argument similar to the ones of Lemma~\ref{le:cost_CP} that $\Gamma$ has $opt(G/C_P)$ outer-fips, that is the cost of $\Gamma'$. By Lemma~\ref{le:cost_CP} we have that  $\Gamma$ is the $d$-dimensional dominance drawing of $G$ having the minimum number of outer-fips and where $C_P$ is compact.  
	
	\smallskip
	For any vertex of $T$ we perform the $O(dk^2(k!)^d)$ time operation that we described in the Base Step and in the Recursive Step. Hence, we have that the algorithm requires  $O(ndk^2(k!)^d)$ time. Computing $T$ requires $O(nm)$. 
\end{proof}

Observe that a constant $d$ is the dominance dimension of a DAG $G=(V,E)$ if $G$ admits a dominance drawing with $d$ dimensions, but not with $d-1$. By Theorem~\ref{th:fpt} we can check in  $O(nm+ndk^2(k!)^d)$ time if $G$ admits a $d$-dimensional dominance drawing (i.e. a weak dominance drawing with $0$ fips). Hence, we have the following corollary of~Theorem~\ref{th:fpt}.

\begin{corollary}
	For any constant $d$, it is possible to test if the dominance dimension of a DAG $G$ is $d$ in $O(nm+ndk^2(k!)^d)$ time, where $k$ is the maximum degree of the vertices of the modular decomposition tree of $G$.
\end{corollary}

\section{Proof of Lemma \ref{le:compaction}} 
\label{se:compaction_proof}

In this section we prove the following lemma, introduced in Section~\ref{se:th1}.
\compact*

Before proving Lemma~\ref{le:compaction} we introduce some notation and prove some intermediate results. Let $M$ be a module of $G$, $\Gamma$ be a weak dominance drawing of $G$, and $D$ be a dimension of $\Gamma$. The \emph{separator} of $M$ in $D$ is a set $S\subseteq V$$\setminus$$M$ so that, for any $v\in S$, there exist two vertices $u,w\in M$ so that $D(u)<D(v)<D(w)$. 

\begin{example} Consider Fig.~\ref{fi:compaction}(a) and the module $M=\{5,7,10,12,14\}$. The separator of $M$ in dimension $X$ and $Y$ is $\{2,3,9,13\}$ and $\{6,8,9,11,16\}$, respectively.
\end{example}
Note that if $S=\emptyset$, $M$ is compact in~$D$. We have the following claim.

\begin{cl}
	\label{le:SMincomparable}
	Let $\Gamma$ be a weak dominance drawing of $G$. Let $M$ be a module if $G$, $D$ be a dimension of $\Gamma$, and $S$ be the separator of $M$ in $D$. Any vertex $u\in S$ is incomparable with the vertices of $M$.
\end{cl}
\begin{proof}
	Let $v\in S$ and $u,w\in M$ such that $D(u)<D(v)<D(w)$. There is no path from $v$ to $u$, since $D(v)>D(u)$. Similarly, there is no path from $w$ to $v$, since $D(w)>D(v)$. Hence, $v$ is incomparable to the vertices of $M$. 
\end{proof}

Let $\Gamma$ be a weak dominance drawing of $G$. In all the following, given a module $M$, we consider the congruence partition $C_P=\{M,V$$\setminus$$M\}$ of $G$. In this setting, a fip $(u,v)$ of $\Gamma$ is an inner-fip of $\Gamma$ if $u,v\in M$ or $u,v\in V$$\setminus$$M$. Otherwise, it is an outer-fip of $\Gamma$. Switching $u$ and $v$ in dimension $D$ is equivalent to setting $D(u)=\alpha$, $D(u)=D(v)$, and $D(v)=\alpha$.  For any $v\in M$, let $out_v$ be the number of outer-fips involving $v$. Let $p$ be a vertex of $M$ with the minimum number of outer fips $out_p$. We now describe an operation that we denote by~\emph{compaction}~of~$M$.

\medskip\noindent \texttt{Compaction ($\Gamma$, $M$):} For every dimension $D$ of $\Gamma$, having separator $S$, perform the following computation: While there are two vertices $u\in S$ and $v\in M$ such that $D(u)=D(v)-1<D(p)$ or $D(p)<D(v)=D(u)-1$, switch~$u$~and~$v$ in $D$. 

\begin{figure}[h]
	\centering
	\subfigure[]{\includegraphics[width=0.3\columnwidth,page=2]{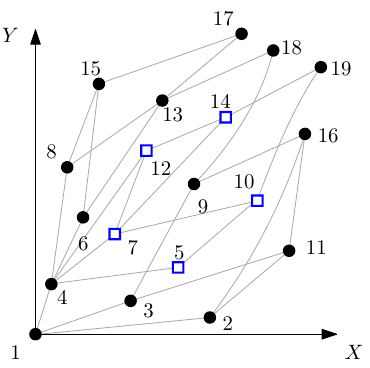}}
	\hfil
	\subfigure[]{\includegraphics[width=0.3\columnwidth,page=3]{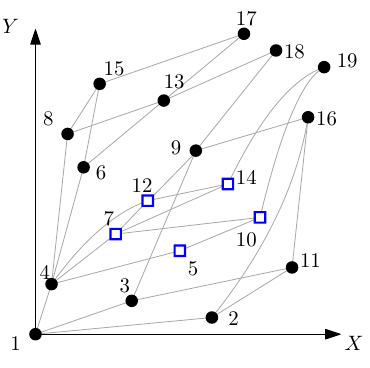}}
	\hfil
	\subfigure[]{\includegraphics[width=0.3\columnwidth,page=4]{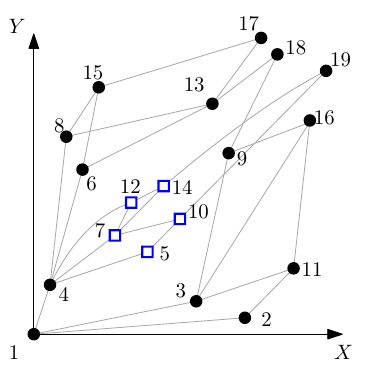}}
	\hfil
	\subfigure[]{\includegraphics[width=0.3\columnwidth,page=5]{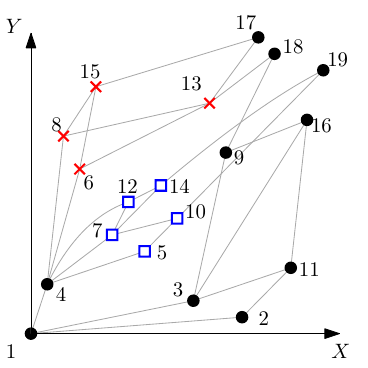}}
	\hfil
	\subfigure[]{\includegraphics[width=0.3\columnwidth,page=6]{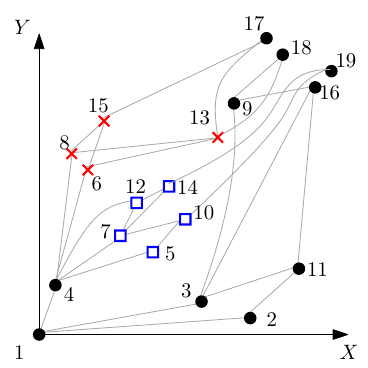}}
	\hfil
	\subfigure[]{\includegraphics[width=0.3\columnwidth,page=7]{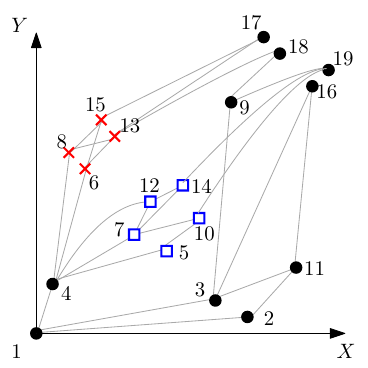}}
	\hfil
	\caption{Illustration of the operation of compaction.}\label{fi:compaction}
\end{figure}

\begin{example}
\label{ex:compaction_outerfip} 
Refer to Fig.~\ref{fi:compaction}(a-c) and the module $M=\{5,7,10,12,14\}$. Consider the 2-dimensional weak dominance drawing in Fig.~\ref{fi:compaction}(a). We have: 

\noindent $\bullet$ $out_5=6$ (outer-fip $(3,5)$ and $(5,w)$ $\forall$ $w\in \{9,11,16,17,18\}$).

\noindent $\bullet$ $out_7=5$ (outer-fip $(7,w)$ $\forall$ $w\in \{9,13,16,17,18\}$).

\noindent $\bullet$ $out_{10}=5$ (outer-fip $(v,10)$ $\forall$ $v\in \{2,3,6\}$) and $(5,w)$ $\forall$ $w\in \{16,18\}$).

\noindent $\bullet$ $out_{12}=7$ (outer-fip $(v,12)$ $\forall$ $v\in \{3,6,8\}$) and $(12,w)$ $\forall$ $w\in \{13,16,17,18\}$).

\noindent $\bullet$ $out_{14}=6$ (outer-fip $(v,14)$ $\forall$ $v\in \{3,6,8,9\}$) and $(12,w)$ $\forall$ $w\in \{17,18\}$).

\smallskip\noindent The vertex $p$ of $M$ having minimum $out_p$ can be either $7$ or $10$. We chose $p=7$. Fig.~\ref{fi:compaction}(b) shows the drawing after that the operation of compaction of $M$ in $\Gamma$ performed its step on dimension $Y$ (and still not on dimension $X$). Fig.~\ref{fi:compaction}(c) shows the graph resulting after the operation of compaction of $M$~in~$\Gamma$. 
\end{example}
By construction, after the compaction of $M$ in $\Gamma$, $M$ is compact~in~$\Gamma$. 

\begin{cl}
	\label{le:compaction-technical}
	Let $M$ be a module of $G$ and let $M'$ be a module different from $M$ that is compact in $\Gamma$. By performing the compaction of $M$, we have that: (1)~$\Gamma$ is (still) a weak dominance drawing; (2)~the number of fips of $\Gamma$ is not increased; (3)~module $M'$ is (still) compact in $\Gamma$.
\end{cl}
\begin{proof}
	 Notice that, after the compaction of $M$,  the relative positions of two vertices $u$ and $v$ in $\Gamma$ changed if and only if $u\in M$ and $v\in S$ or vice versa. This fact has two implications. First, by Claim~\ref{le:SMincomparable}, $\Gamma$ remains a weak dominance drawing and we have that Property~(1) is verified. Second, we have that:
	 
	 \smallskip\noindent (i) The number of inner-fips of $\Gamma$ does not change (recall $S\subseteq V$$\setminus$$M$).
	 
	 \smallskip\noindent
	 By construction, the relative position of any $u\in G$ and $p$ do not change. Also, recall that $M$ is compact in $\Gamma$. Hence:
	 
	 \smallskip\noindent(ii) The number of fips of $\Gamma$ involving $p$ does not change.
	 
	\smallskip\noindent(iii) The relative position between any $v\in V$$\setminus$$M$ and $p$ is the same as the relative position between $v$ and any  $u\in M$.
	
	\smallskip
	Consideration~(iii) implies that every vertex $v\in M$ is involved in $out_p$ outer-fips of $M$ and, by~(ii) and since $out_p\le out_v$ before the compaction, we have that the number of outer-fips involving $v$ is not augmented. Hence, by~(i), the number of fips of $\Gamma$ is not increased. It follows that Property~(2) of the claim is verified. It remains to show Property~(3). Let $D$ be any dimension of $\Gamma$ and $S$ be the separator of $M$ in $D$. If $M'\cap S= \emptyset$ the compaction of $M$ does not modify the position of the vertices of $M'$ and $M'$ remains compact. Suppose  $M'\cap S\not =\emptyset$. Since $M'$ is compact we have $M'\subseteq S$. Hence, if we switch $u\in M'$ and $v\in M$, then we switch $v$ with all the vertices of $M'$. Hence, $M'$ remains compact. 
\end{proof}
See Example~\ref{ex:propertyclaimcomptech} in the Appendix for an illustration of the three properties of Claim~\ref{le:compaction-technical}. We now prove Lemma \ref{le:compaction}. Let $\Gamma$ be a cost-minimum weak dominance drawing of $G$. We show that, given $\Gamma$, it is possible to compute a cost-minimum weak dominance drawing $\overline{\Gamma}$ of $G$ having the same number of fips of $\Gamma$ and where $C_P$ is compact. We initialize $\overline{\Gamma}=\Gamma$. For every $M\in C_P$ we perform the compaction of $M$ in $\overline{\Gamma}$. After every compaction we have that: $M$ is compact in $\overline{\Gamma}$ by construction; $\overline{\Gamma}$ is still a weak dominance drawing by Property~(1) of Claim~\ref{le:compaction-technical}; $\overline{\Gamma}$ has the minimum number of fips by  Property~(2) of Claim~\ref{le:compaction-technical}; every module $M'$ the that was compact before the compaction of $M$ remains compact by  Property~(3) of Claim~\ref{le:compaction-technical}. Hence, $\overline{\Gamma}$ is a weak dominance drawing of $G$ with the minimum number of fips and where $C_P$ is compact. This proves Lemma~\ref{le:compaction}.

\section{Concluding Remarks}
In this paper we present a fixed parameter algorithm solving the fips-minimization problem, which is NP-hard. We show that if the maximum degree of the vertices of the modular decomposition tree of a DAG $G$ is a constant, then the problem is polynomial-time solvable for any constant number of dimensions. 
We use a brute force algorithm to obtain cost-minimum weak dominance drawing of the quotient graphs associated to every module of the modular decomposition~tree.

Observe that the additive term ``$nm$'' in the complexity of  Theorem~\ref{th:fpt} is required to compute $G^*$ and consequently the path-based modular decomposition tree $T$ of $G$. Our results hold if we use the concept of edge-based module instead of the concept of path-based module. In this case, since computing the edge-based modular decomposition requires $O(m)$ time, the time complexity of Theorem~\ref{th:fpt} is $O(m+ndk^2(k!)^d)$. However, as we already observed in Section~\ref{se:preliminaries}, $k$ for edge-based modules is typically greater than $k$ for path-based modules.

Notice that in general the concept of module can be defined in many ways. For example, in~\cite{DBLP:journals/dase/AnirbanWI19} they consider only path-based modules $M$ where $G_M$ is a path or a set of $|M|$ incomparable vertices. Given a definition of module, the time complexity of  Theorem~\ref{th:fpt} is $O(y+ndk^2(k!)^d)$, where $m \le y\le nm$ is the time complexity to compute the modular decomposition tree. However, the more restrictive the definition of module, the higher the value of $k$.

\paragraph{Open problems:} It would be interesting to find polynomial time heuristics to minimize this cost and to run experiments using these heuristics instead of the brute force approach in order to find out if they compute a smaller number of fips than is known in the literature~\cite{DBLP:journals/www/LiHZ17,DBLP:conf/edbt/VelosoCJZ14}. Another interesting problem is to find algorithms computing drawings not with the minimum number of fips, but with a bounded number of fips. In this case, the problem of computing all the fips efficiently could become very interesting also in practice. Finally, we believe that also computing weak dominance drawings where the number of vertices involved in fips is minimized could be an important step forward in this line of research.

\clearpage
\bibliographystyle{plain}
\bibliography{Literature}

\clearpage
\section*{Appendix}

\begin{example} 
	\label{ex:propertyclaimcomptech}
	Illustration for the three properties of Claim~\ref{le:compaction-technical} ($d=2$).
	
	\noindent
	\texttt{- Property~(1).} The drawing in Fig.~\ref{fi:compaction}(c), obtained by performing the compaction on module $M=\{5,7,10,12,14\}$ in the weak dominance drawing in Fig.~\ref{fi:compaction}(a), is a weak dominance drawing. 
	
	\noindent
	\texttt{- Property~(2).} The outer-fips involving the any vertex $v\in M$ in the weak dominance drawing in Fig.~\ref{fi:compaction}(c) are  $(3,v)$ and $(v,w)$ $\forall$ $w\in \{9,11,16,17,18\}$. Before the compaction, Fig.~\ref{fi:compaction}(a), any vertex $v$ where involved in not less outer-fips (see Example~\ref{ex:compaction_outerfip}). Notice that vertex $p=7$ is involved in the same fips in both drawings. The inner-fips in Fig.~\ref{fi:compaction}(a) and~(c) are the same. Hence, the drawing in Fig.~\ref{fi:compaction}(c) has no more fips than the one in Fig.~\ref{fi:compaction}(a).
	
	\noindent
	\texttt{- Property~(3).}  Refer to Fig.~\ref{fi:compaction}(d-f). Denote now $M=\{6,8,13,15\}$ and $M'=\{5,7,10,12,14\}$. Consider the weak dominance drawing in Fig.~\ref{fi:compaction}(d). Fig.~\ref{fi:compaction}(e) shows the drawing after that the operation of compaction of $M$ in $\Gamma$ performed its step on dimension $Y$ (and still not on dimension $X$). Fig.~\ref{fi:compaction}(f) shows the graph resulting after the operation of compaction of $M$~in~$\Gamma$. Notice that Module $M'$, that is compact in Fig.~\ref{fi:compaction}(d), is still compact in Figures~\ref{fi:compaction}(e) and~(f).
\end{example}

\end{document}